\documentclass[journal]{IEEEtran}
\usepackage{algorithmic}
\usepackage{array}
\usepackage{stfloats}
\usepackage{url}
\usepackage{graphicx}
\usepackage{enumerate}
\usepackage{amsthm,amsmath,amssymb}
\usepackage{mathrsfs}
\usepackage{subfigure}  
\usepackage{epstopdf}
\usepackage{epsfig}%
\usepackage{cite}

\begin{document}
\title{Modular WSS-based OXCs for Large-Scale Optical Networks}

\author{Tong~Ye,~\IEEEmembership{Member,~IEEE,}
        Kui~Chen~
\thanks{This work was supported by the National Science Foundation of China (61671286).}
\thanks{Tong Ye and Kui Chen are with the State Key Laboratory of Advanced Optical Communication Systems and Networks, Shanghai Jiao Tong University, Shanghai 200240, China. (email: \{yetong, aogutiancheng\}@sjtu.edu.cn).}
}
\maketitle
\begin{abstract}
The explosive growth of broadband applications calls for large-scale optical cross-connects (OXCs). However, the classical wavelength selective switch (WSS) based OXC is not scalable in terms of the size of employed WSSs and the cabling complexity. To solve this problem, we propose a three-phase approach to construct a modular WSS-based OXC. In phase 1, we factorize the interconnection network between the input stage and the output stage of the traditional OXC into a set of small-size interconnection networks. In phase 2, we decompose each WSS into a two-stage cascaded structure of small-size WSSs. In phase 3, we combine the small-size interconnection networks with the small-size WSSs to form a set of small-size OXC modules. At last, we obtain a modular OXC, which is a network of small-size OXCs. Similar to the classical OXC, the modular OXC is nonblocking at each wavelength and possesses a self-routing property. Our analysis shows that the modular OXC has small cabling complexity and acceptable physical-layer performance.
\end{abstract}
\begin{IEEEkeywords}
Optical Cross-connect (OXC), Wavelength Selective Switch (WSS), Modularity, Cabling Complexity.
\end{IEEEkeywords}
\IEEEpeerreviewmaketitle
\section{Introduction}\label{introduction}
\IEEEPARstart{I}{n} recent years, optical circuit switches (OCSs) were gradually applied to metro area networks and backbone networks to provide huge capacity. With the advent of 5G era, ubiquitous broadband applications, such as video sharing, spring up extensively. A single optical fiber with wavelength division multiplexing (WDM) technology may not be enough to carry the traffic between two adjacent nodes. Usage of multiple fibers in a network link at the same time has been a development trend of the next-generation network \cite{01}. According to the prediction of \cite{01}, one link may need to use up to 80 parallel fibers in 2027. As a result, large-scale OCSs are highly desirable in the near future.

However, as the key component of OCSs, traditional optical cross-connects (OXCs) cannot scale to a large port count. The function of OXCs is to provision connections in a nonblocking manner. As Fig. \ref{fig1} shows, each port carries multiple wavelengths. The OXC can always build up a connection from an input to an output, if they have the same idle wavelength. An example is the connection at wavelength $\lambda_1$ from input 3 to output 2 in Fig. \ref{fig1}. The traditional OXC is not scalable due to the following two reasons. In a classical $N\times N$ OXC, there are $N$ $1\times N$ wavelength selective switches (WSSs) at the input stage and $N$ $N\times1$ WSSs at the output stage. Currently, the port count of the WSS is only $\sim20$, which limits the port count of the OXC to $\sim20$. Also, the OXC connects an input WSS to an output WSS via exactly one fiber, and thus contains $N^2$ fibers. When $N$ is large, the number of fibers will be very large and thus lead to a high cabling complexity. As \cite{02} pointed out, high cabling complexity has become the development obstruction of the OCS.
\begin{figure}
\centering
\includegraphics[scale=0.6]{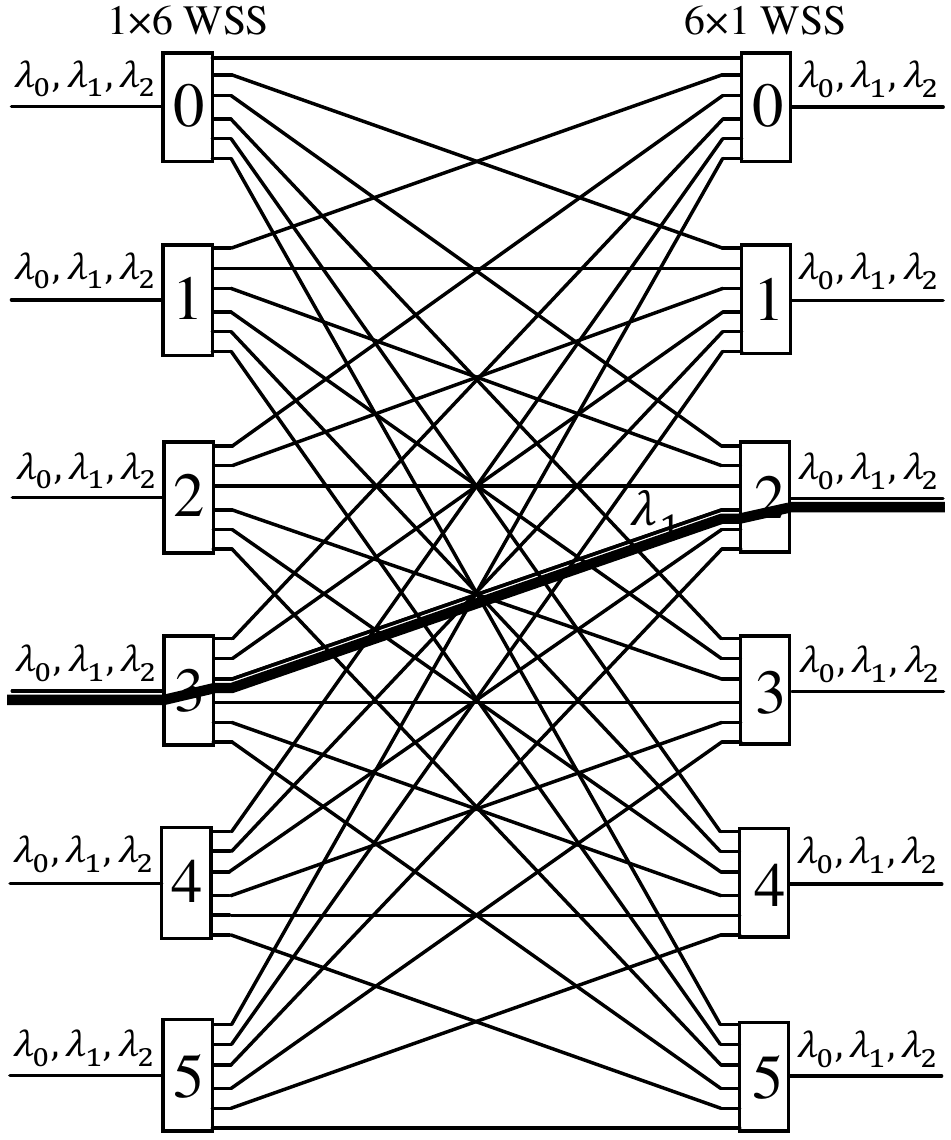}
\caption{A classical $N\times N$ WSS-based OXC.}
\label{fig1}
\end{figure}
\subsection{Previous Work}
To address this problem, several previous works have been focused on the construction of large-scale OXCs \cite{03,05,06,07,08,09,10,11}.

The first one is HIER proposed by \cite{03,05}. The port count of HIER is $N=DF$, where $D$ is the degree of the node, i.e., the number of links that connect to other nodes, and $F$ is the number of fibers on each link. Ref. \cite{03} first decomposes a $1\times N$ WSS to a $1\times D$ WSS, each output of which is attached by a $1\times F$ WSS, and then replaces each $1\times F$ WSS by a wavelength non-sensitive $1\times F$ optical space-switching module. The motivation of such replacement is to reduce the system cost at the expense of the flexibility of wavelength switching. Thus, HIER is internally blocking. As \cite{03} shows, the optical network equipped with HIERs should take into consideration the internal structure and the switching state of the nodes when it sets up an end-to-end lightpath for a request. This may impose a high computation complexity on the routing process. Also, at the output side of HIER, there are $D$ $N\times F$ delivery and coupling switch (DCSW) modules, each of which includes $NF$ fibers. Typically, $D$ is a small integer, and thus $F$ is on the same order of $N$. This implies that each DCSW module is a large-scale module if $N$ is large. Another issue of HIER is that, each output port of the DCSW is an $N\times1$ optical coupler (OC), which will lead to a large insertion loss when $N$ is large.

To cut down the cost and the insertion loss of HIER, \cite{07,08} proposed a two-stage OXC. Different from HIER, this two-stage structure replaces each $1\times F$ WSS by a $1\times F$ waveband selective switch (WBSS), and each $N\times1$ OC by a cascaded structure including $M$ $(N/M)\times1$ OCs and an $M\times1$ WSS. Though the two-stage OXC is more flexible than HIER, it is also internally blocking.

Ref. \cite{09,10} proposed a modular OXC that is constructed by interconnecting a set of small-size OXCs in a ring topology. In particular, most of the ports of each small-size OXC are the ports of the modular OXC, while the remaining ports are connected with other small-size OXCs. This design is modular since it only employs small-size OXCs. Another advantage is that it has a low cabling complexity. However, it is clear that this structure is internally blocking. Also, the insertion loss between different input-output pair is quite different. A connection that has to pass through a number of small-size OXCs will suffer a large optical loss, which limits the scalability of this design.

Ref. \cite{11} presented an arrayed waveguide grating (AWG) based OXC for high-performance computers. Different from that of the traditional OXC, each port of this OXC only carries one wavelength at the same time. In an $N\times N$ AWG-based OXC, there are $N$ $1\times K$ optical switches (OSs) at the input stage and $K$ $M\times M$ AWGs at the output stage, where $N=2KM$. The signal of a connection is generated by a tunable laser. The OXC selects an AWG for a connection via a $1\times K$ OS and determines the output of the AWG by tuning the wavelength of the laser. This AWG-based OXC can route the connections in a nonblocking manner. However, the cabling complexity is high. Note that a classical OXC can carry $Nw$ connections at the same time. To provide the same capacity, we need an $Nw\times Nw$ AWG-based OXC, which contains $(K+1/2)Nw$ fibers. If $K=\sqrt{Nw}$, the cabling complexity is $O\left(N^{1.5}w^{1.5}\right)$.

In summary, the existing designs for large-scale OXCs are either internally blocking or not scalable in terms of cabling complexity.
\subsection{Our Contributions}
This paper proposes a method to construct a large-scale OXC, which is scalable and nonblocking at each wavelength. This method is essentially a modularization process of traditional OXCs. There are three phases in this modularization process. In phase 1, we factorize the interconnection network between the input stage and the output stage of the traditional OXC into a set of small-size interconnection networks. In phase 2, we decompose each WSS into a two-stage cascaded structure of small-size WSSs. In phase 3, we combine the small-size interconnection networks with the small-size WSSs to form a set of small-size OXCs. At last, we obtain a modular OXC, which is actually a network of small-size OXCs.

The proposed $N\times N$ OXC possesses the following advantages: 
\begin{enumerate}[(1)]
\item It is modular since it can be constructed from a set of $1\times n$ WSSs and $r\times r$ OXC modules, where $N=nr$; 
\item It is internally nonblocking at each wavelength, and possesses the same routing process as the traditional OXC;
\item Its cabling complexity is small because the number of fiber links is $Nn$, only $2/r$ of that of the traditional OXC;
\item Its physical-layer transmission performance and insertion loss are acceptable, meaning that it is feasible for practical applications.
\end{enumerate}

The rest paper is organized as follows. In Section \ref{Preliminary}, we introduce the structure and property of the traditional OXCs. We propose a three-phase process to modularize the traditional OXC and finally obtain the modular OXC in Section \ref{construction}. We evaluate the performance of the modular OXC in terms of scalability, routing property, and physical-layer transmission performance in Section \ref{evaluation}. Section \ref{conclusion} concludes this paper.

\section{Preliminary} \label{Preliminary}
This section introduces the structure and the function of the classical OXC, since it is the start point of our modular design. As Fig. \ref{fig1} shows, an $N\times N$ OXC consists of $N$ $1\times N$ WSSs and $N$ $N\times 1$ WSSs with $N^2$ fibers in between. Each input (output) of the OXC carries $w$ wavelengths $\lambda_0,\lambda_1,\cdots,\lambda_{w-1}$. We denote an $N\times N$ OXC associated with $w$ wavelengths as $\mathcal{Q}\left(N,w\right)$, where the input WSSs are labelled by $0,1,\cdots,p,\cdots,N-1$ and the output WSSs are numbered by $0,1,\cdots,q,\cdots,N-1$ from top to bottom. As an example, the $6\times6$ OXC in Fig. \ref{fig1} is denoted by $\mathcal{Q}\left(6,3\right)$.
\subsection{$1\times N$ WSS}
There are several kinds of $1\times N$ WSSs, such as micro-electro mechanical system (MEMS) based WSS \cite{12}, liquid crystal on silicon (LCOS) based WSS \cite{13}, and planar lightwave circuit (PLC) based WSS \cite{14}. Herein, we take the PLC-based WSS as an example to elaborate the function of the $1\times N$ WSS.
\begin{figure}[bp]
\centering
\subfigure[]{
 \label{fig2a}
 \includegraphics[scale=0.8]{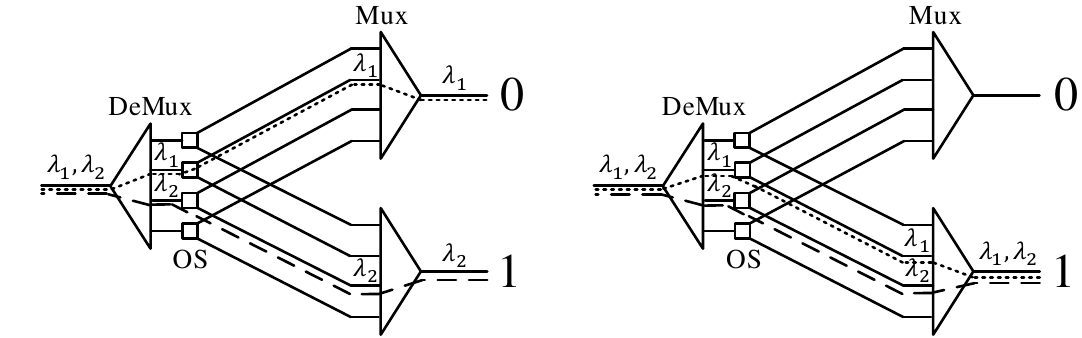}}
\subfigure[]{
 \label{fig2b}
 \includegraphics[scale=0.8]{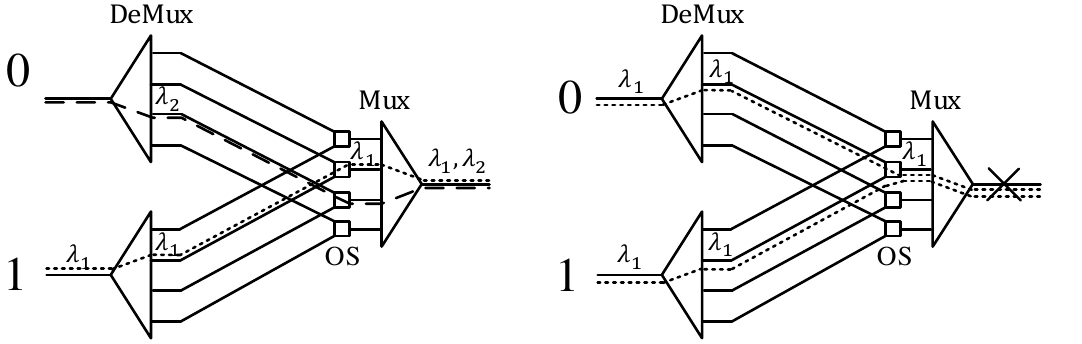}}
\caption{Logical structure of (a) a $1\times2$ WSS and (b) a $2\times1$ WSS.}
\label{fig2}
\end{figure}

The function of a $1\times N$ WSS is to switch wavelengths to their desired outputs. In a PLC-based $1\times N$ WSS, there are a $1\times w$ wavelength demultiplexers (DeMuxs) at the input stage, $w$ $1\times N$ OSs in the middle stage, and $N$ $w\times1$ wavelength multiplexers (Muxs) at the output stage \cite{14}. The input WDM signal is demultiplexed by the $1\times w$ DeMux. Each demultiplexed wavelength is then fed to a $1\times N$ OS attached at each output of the $1\times w$ DeMux. Since each output of $1\times N$ OS connects with an input of a $w\times1$ Mux, it can switch the demultiplexed wavelength to the desired $w\times1$ Mux. In other words, the $1\times N$ WSS switches each wavelength separately by a dedicated $1\times N$ OS such that the switching of different wavelengths is mutually independent \cite{14}. Fig. \ref{fig2a} plots a $1\times2$ PLC-based WSS associated with 4 wavelengths. It can either switch $\lambda_1$ and $\lambda_2$ to outputs 0 and 1, respectively, or route both of them to output 1.

The $N\times1$ WSS is the inverse of a $1\times N$ WSS. As long as no two signals at the same wavelength are fed to the $N\times1$ WSS, it can multiplex the wavelengths from different inputs to the output. At each input, the WDM signal is first demultiplexed and each demultiplexed wavelength is then separately switched by a dedicated $N\times1$ OS to the $N\times1$ Mux. Fig. \ref{fig2b} illustrates the cases with and without contention.
\subsection{Shuffle Interconnection}
The $N^2$ fibers in OXC $\mathcal{Q}\left(N,w\right)$ is actually an $N^{2}\times N^{2}$ interconnection network. In the following, we show that this interconnection network is actually an $N^{2}\times N^{2}$ shuffle network \cite{15}.

Consider an $N^{2}\times N^{2}$ interconnection network, where $N^2$ total inputs can be divided into $N$ groups, each of which contains $N$ inputs, and $N^2$ total outputs can be divided into $N$ groups, each of which contains $N$ outputs.
\newtheorem{definition}{Definition}
\begin{definition}\label{definition1}
An $N^2\times N^2$ interconnection network is an $N^2\times N^2$ shuffle network, denoted as $\mathcal{S}(N)$, if the $q$th input of the $p$th input group connects to the $p$th output of the $q$th output group, where $p$, $q=0,1,\cdots,N-1$.
\end{definition}
In $\mathcal{Q}\left(N,w\right)$, the $N$ outputs of each $1\times N$ WSS and the $N$ inputs of each $N\times1$ WSS naturally constitute an input group and an output group, respectively. In particular, the $q$th output of the $p$th $1\times N$ WSS connects with the $p$th input of the $q$th $N\times1$ WSS. It follows that these $N^2$ fibers is an $N^2\times N^2$ shuffle network. For example, the 36 fibers in OXC $\mathcal{Q}\left(6,3\right)$ in Fig. \ref{fig1} is a $36\times36$ shuffle network $\mathcal{S}(6)$, as Fig. \ref{fig3} shows.

We use two bits to label an input (output) of $\mathcal{S}(N)$, where the first bit and the second bit are called group address and port address, respectively. The $q$th input of the $p$th input group is labelled as $pq$ and the $p$th output of the $q$th output group is labelled as $qp$, where $p$, $q=0,1,\cdots,N-1$. With such a numbering scheme, a shuffle network possesses the following properties:
\begin{enumerate}[P1.]
\item These is exactly one fiber between an input group and an output group;
\item If an input connects with an output, the bit address of the output is the exchange of the two sub-addresses of the input.
\end{enumerate}
\noindent That is, input $pq$ connects with output $qp$. We denote the fiber between input $pq$ and output $qp$ as $f\left(pq,qp\right)$. As Fig. \ref{fig3} illustrates, fiber $f(32,23)$ links input 32 to output 23.

We say an $N^2\times N^2$ interconnection network is functionally equivalent to an $N^2\times N^2$ shuffle network if it has the features P1 and P2.
\begin{figure}
\centering
\includegraphics[scale=0.7]{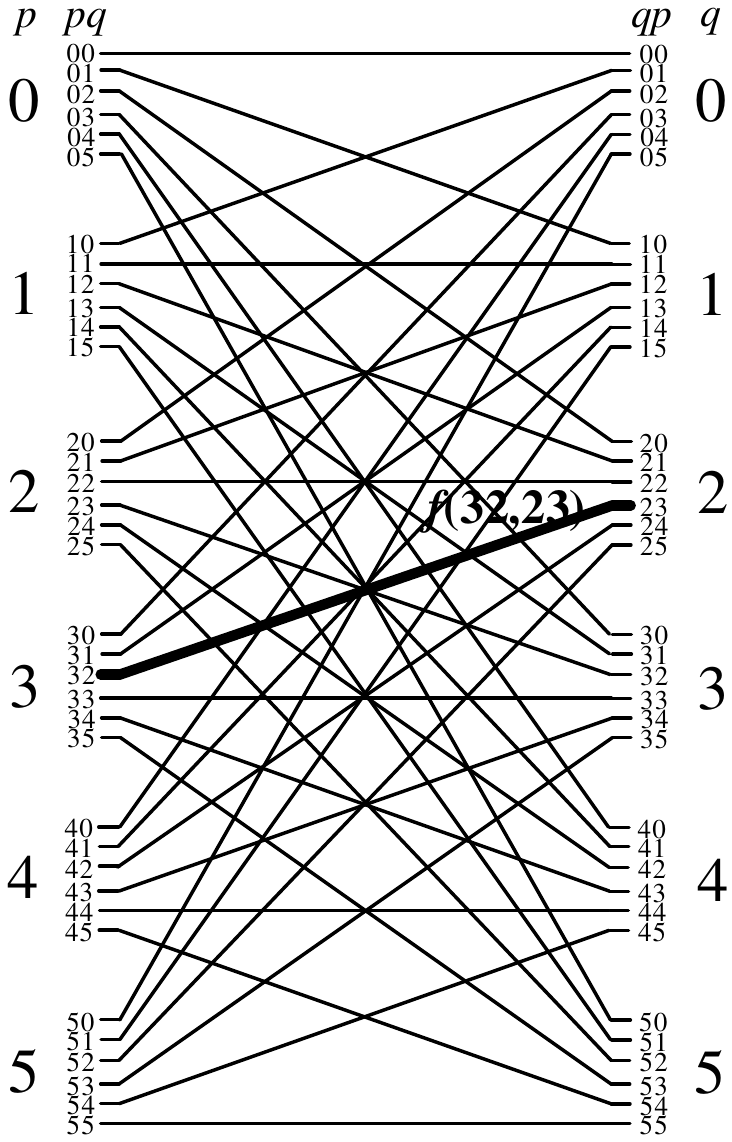}
\caption{A $36\times36$ shuffle network $\mathcal{S}(6)$.}
\label{fig3}
\end{figure}

The interconnection pattern of $\mathcal{S}(N)$ can be delineated by an $N\times N$ connectivity table, denoted by $T$. In the table, the $p$th row and the $q$th column are respectively corresponding to the $p$th input group and the $q$th output group of $\mathcal{S}\left(N\right)$. We denote the intersection of row $p$ and column $q$ as $\left(p,q\right)$, which is corresponding to the fiber $f(pq,qp)$ connecting input $pq$ to output $qp$. We thus fill entry $\left(p,q\right)$ with $pq$, $qp$. As an example, the connectivity table $T$ is given by Fig. \ref{fig4}, where entry $\left(3,2\right)$ is corresponding to fiber $f(32,23)$.
\begin{figure}[bp]
\centering
\includegraphics[scale=0.83]{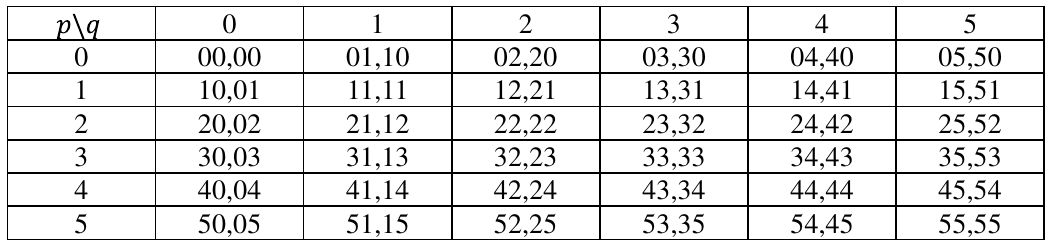}
\caption{Connectivity table $T$ of a $36\times36$ shuffle network $\mathcal{S}(6)$.}
\label{fig4}
\end{figure}
\subsection{Routing Property of OXC $\mathcal{Q}\left(N,w\right)$}
The function of the OXC is to setup a connection at a wavelength from an input to an output. We denote a connection from wavelength $\lambda_i$ at input $p$ to $\lambda_i$ at output $q$ as $R\left(p,q,\lambda_i\right)$, where $i=0,1,\cdots,w-1$. According to the structure of $\mathcal{S}(N)$, $\mathcal{Q}\left(N,w\right)$ has a self-routing property, that is, the path of $R$ is uniquely determined by its input address $p$ and its output address $q$ as follows:
\begin{align}
&\text{ Input of }1\times N \text{ WSS } p \nonumber \\
\rightarrow& \text{ fiber } f(pq,qp) \text{ of } \mathcal{S}(N) \nonumber \\
\rightarrow& \text{ output of } N\times1 \text{ WSS } q. \nonumber
\end{align}

It is clear that the connections at different wavelengths will not conflict with each other, since the $1\times N$ WSS or the $N\times1$ WSS can switch different wavelengths independently. In the following, we show that the OXC can also route the connections at the same wavelength in a nonblocking manner.
\newtheorem{lemma}{Lemma}
\begin{lemma} \label{lemma1}
OXC $\mathcal{Q}\left(N,w\right)$ is nonblocking at each wavelength.
\end{lemma}
\begin{proof}
Consider an extreme case, where wavelength $\lambda$ is idle only at input $p$ and output $q$. In the following, we show that a connection $R\left(p,q,\lambda\right)$ can still be established in this case.

Wavelength $\lambda$ is free at input $p$, and thus is idle at all the outputs of $1\times N$ WSS $p$. Similarly, $\lambda$ is idle at all the inputs of $N\times1$ WSS. It follows that $\lambda$ is available at fiber $f(pq,qp)$. We thus can build up a connection from input $p$ to output $q$ via wavelength $\lambda$.
\end{proof}
\section{Construction Of Modular OXCs} \label{construction}
As Section \ref{introduction} states, the number of fibers in the classical OXC is the square of the port count of the OXC, and thus increases very fast with the size of OXC. We thus start the modularization of the classical OXC with the factorization of the interconnection network. Specifically, we devise a three-phase modularization process: (1) factorize the $N^2\times N^2$ shuffle network into a set of small-size shuffle networks, (2) decompose each WSS to a cascaded structure of small-size WSSs, and (3) merge the decomposed WSSs with the factorized shuffle networks to form small-size OXCs. In this way, we obtain a modular OXC, which is actually an interconnection network of small-size OXCs.
\subsection{Shuffle Network Factorization}
In this part, we construct an $N^2\times N^2$ modular shuffle network, denoted by $\hat{\mathcal{S}}\left(n\times r\right)$, using a set of $r^2\times r^2$ shuffle networks $\mathcal{S}(r)$, where $N=nr$ and $n=1,2,\cdots$. In the target shuffle network $\hat{\mathcal{S}}\left(n\times r\right)$, there are $nr$ input groups, each containing $nr$ inputs, and $nr$ output groups, each including $nr$ outputs. The port count of $\hat{\mathcal{S}}\left(n\times r\right)$ is $n^2$ times that of  $\mathcal{S}(r)$. If we use $n^2$ $r^2\times r^2$  $\mathcal{S}(r)$s to construct $\hat{\mathcal{S}}\left(n\times r\right)$, the inputs and the outputs of each $r^2\times r^2$ shuffle network  $\mathcal{S}(r)$ should be the inputs and the outputs of $\hat{\mathcal{S}}\left(n\times r\right)$, respectively. Also, every $nr$ inputs (outputs) of $r^2\times r^2$ shuffle network  $\mathcal{S}(r)$ should be grouped together as an input (output) group. To achieve this goal, we use an approach similar to that in \cite{16}, which factorizes the AWG-based shuffle networks with the help of connectivity table.

We denote the connectivity table of $\hat{\mathcal{S}}\left(n\times r\right)$ as $\hat{T}$, where each row and each column are corresponding to an input group and an output group of $\hat{\mathcal{S}}\left(n\times r\right)$, respectively. According to properties of shuffle networks and the above discussion, a legitimate table $\hat{T}$ should satisfy the following conditions:
\begin{enumerate}[R1]
\item Each entry represents one and only one fiber connecting a pair of input group and output group;
\item Each entry contains two addresses, and one address is the inverse of the two sub-address of the other one;
\item $\hat{T}$ can be divided to $n^2$ sub-tables, each of which is the connectivity table of an $r^2\times r^2$ shuffle network $\mathcal{S}(r)$.
\end{enumerate}
\begin{figure}[bp]
\centering
\includegraphics[scale=0.78]{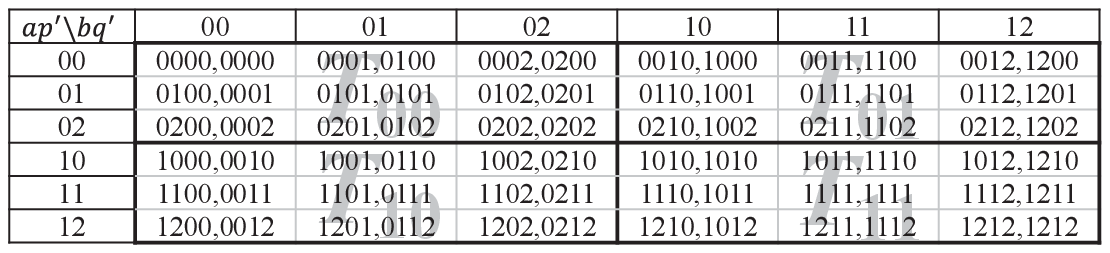}
\caption{Connectivity table $\hat{T}$ of a $36\times36$ modular shuffle network $\hat{\mathcal{S}}\left(2\times3\right)$.}
\label{fig5}
\end{figure}

In the following, we show that such a table $\hat{T}$ can be derived from the connectivity table $T$ of an $N^2\times N^2$ shuffle network $\mathcal{S}\left(n\times r\right)$. Recall that $T$ is an $N\times N$ table. We first apply a modulo$-r$ operation to each number in $T$. The row index $p$ and the column index $q$ change to $p^\prime=[p]_r$ and $q'=[q]_r$, respectively, where $[X]_Y\triangleq X$ mod $Y$, $p,q=0,1,\cdots,N-1$ and $p',q'=0,1,\cdots,r-1$. Accordingly, entry $\left(p,q\right)$ varies from $pq$, $qp$ to $p^\prime q^\prime$, $q' p'$. Thus, in the table after the modulo$-r$ operation, the row index $p'$ periodically increases from 0 to $r-1$, and there are $n$ periods in total. The same thing happens with the column index $q'$. Every $r$ rows and $r$ columns in a period form a sub-table, which is exactly the connectivity table of an $r^2\times r^2$ shuffle network $\mathcal{S}(r)$. In other words, the new table contains $n^2$ identical $r\times r$ sub-tables. To distinguish these sub-tables, we add one bit, denoted by $a$, before the row index $p'$, and one bit, denoted by $b$, before the column index $q'$, where $a,b=0,1,\cdots,n-1$. After these operations, Table $T$ has the following changes:
\begin{enumerate}[a)]
\item Row $p$ changes to row $ap'$ since it is the $p'$th row in the $a$th period, where $p=ar+p'$;
\item Column $q$ changes to column $bq'$ since it is the $q'$th column in the $b$th period, where $q=br+q'$;
\item The content in entry $\left(p,q\right)$ changes to $ap'bq'$, $bq'ap'$.
\end{enumerate}
\noindent This yields a new table, denoted by $\hat{T}$. Also, we denote the sub-table identified by $a$ and $b$ as $T_{ab}$. For example, after applying the modulo operations to table $T$ in Fig. \ref{fig4}, we obtain the connectivity table in Fig. \ref{fig5}. It is very easy to check that table $\hat{T}$ satisfies requirements R1, R2 and R3.
\begin{figure}[tp]
\centering
\includegraphics[scale=0.65]{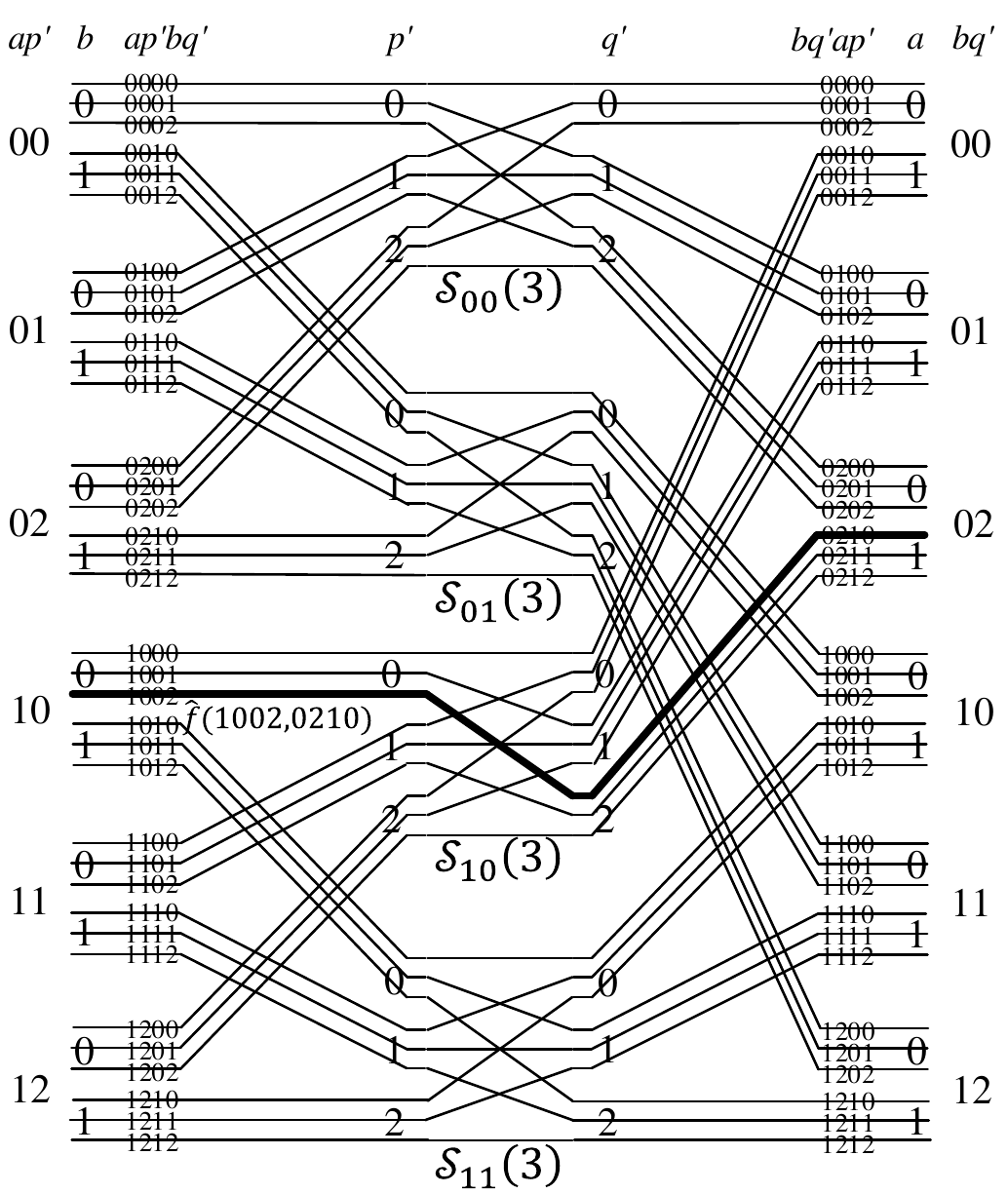}
\caption{A $36\times36$ modular shuffle network $\hat{\mathcal{S}}\left(2\times3\right)$.}
\label{fig6}
\end{figure}

According to table $\hat{T}$, we can construct modular shuffle network $\hat{\mathcal{S}}\left(n\times r\right)$ as follows:
\begin{enumerate}[a)]
\item For row $ap'$, create an input group $ap'$, which includes $nr$ inputs, and divide this group into $n$ subgroups, each of which contains $r$ inputs, where the $q'$th input in the $b$th subgroup is labelled as $ap'bq'$;
\item For column $bq'$, create an output group $bq'$, which includes $nr$ outputs, and divide this group into $n$ subgroups, each of which contains $r$ outputs, where the $p'$th output in the $a$th subgroup is labelled as $bq'ap'$;
\item Layout $n^2$ $r^2\times r^2$ shuffle sub-networks vertically, where the $k$th sub-network is labelled as $\mathcal{S}_{ab}(r)$ if $k=an+b$;
\item Link the ports of input subgroup $ap'b$ to the ports of input group $p'$ of $\mathcal{S}_{ab}(r)$, since the $b$th period of the $ap'$th row in $\hat{T}$ is the $p'$th row of $T_{ab}$;
\item Link the ports of output subgroup $bq'a$ to the ports of output group $q'$ of $\mathcal{S}_{ab}(r)$, since the $a$th period of the $bq'$th column in $\hat{T}$ is the $q'$th column of $T_{ab}$.
\end{enumerate}

The network $\hat{\mathcal{S}}\left(2\times3\right)$ constructed from table $\hat{T}$ in Fig. \ref{fig5} is plotted in Fig. \ref{fig6}. We show that table $\hat{T}$ is indeed the connectivity table of $\hat{\mathcal{S}}\left(2\times3\right)$. Consider input group $ap'=10$ and output group $bq'=02$ of $\hat{\mathcal{S}}\left(2\times3\right)$. Input group 10 connects with sub-networks $\mathcal{S}_{10}(3)$ and $\mathcal{S}_{11}(3)$, and output group 02 connects with $\mathcal{S}_{00}(3)$ and $\mathcal{S}_{10}(3)$. Thus, the fiber between input group 10 and output group 02 must pass through $\mathcal{S}_{10}(3)$. That is, the trace of the fiber is as follows:
\begin{align}
&\text{ Input subgroup } ap'b=100 \nonumber \\
\rightarrow& \text{ input group }p'=0 \text{ of } \mathcal{S}_{ab} (r)=\mathcal{S}_{10} (3) \nonumber \\
\rightarrow& \text{ output group } q'=2 \text{ of } \mathcal{S}_{ab} (r)=\mathcal{S}_{10} (3) \nonumber \\
\rightarrow& \text{ output subgroup } bq' a=021. \nonumber
\end{align}
As Fig. \ref{fig6} plots, port $q'=2$ of input subgroup $ap' b=100$ (i.e., input $ap'bq'=1002$) connects with input $p' q'=02$ of $\mathcal{S}_{10} (3)$, and port $p'=0$ of output subgroup $bq'a=021$ (i.e., output $ap'bq'=0210$) connects with output $q'p'=20$ of $\mathcal{S}_{10}(3)$. According to property P2 of shuffle networks, input $p'q'=02$ of $\mathcal{S}_{10}(3)$ connects with output $q'p'=20$. Thus, input $ap'bq'=1002$ connects with output $bq'ap'=0210$. Accordingly, the content in entry $(10,02)$ of table $\hat{T}$ in Fig. \ref{fig5} is 1002, 0210. This example clearly shows that there is only one fiber, denoted as $\hat{f}\left(ap'bq',bq'ap'\right)$, connecting input group $ap'$ and output group $bq'$, and the input address $ap'bq'$ is the inverse of the two sub-addresses of the output address $bq'ap'$. This indicates that $\hat{\mathcal{S}}\left(n\times r\right)$ has the following property.
\begin{figure}[tp]
\centering
\includegraphics[scale=0.65]{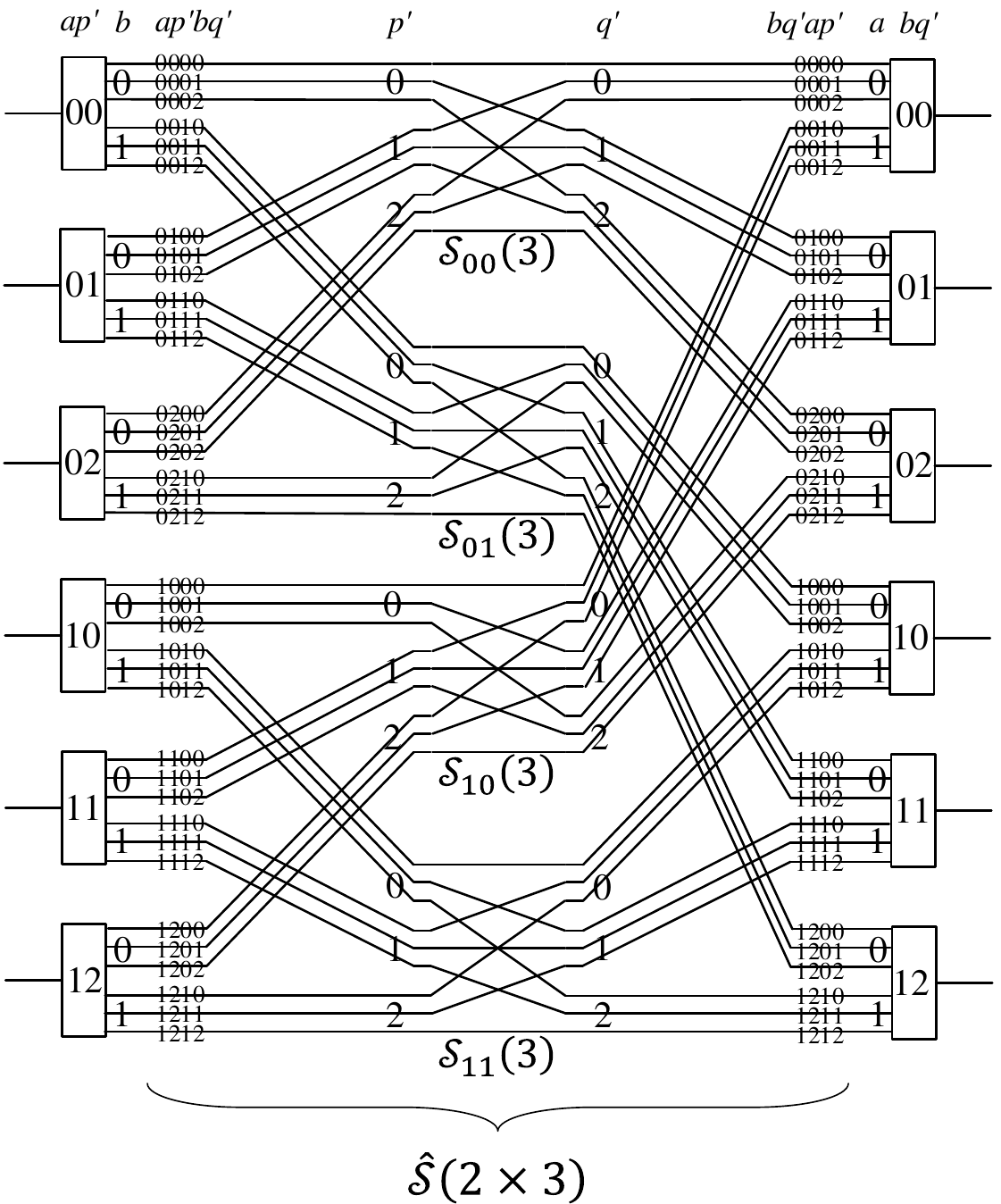}
\caption{A $6\times6$ switching network $Q'\left(2\times3,3\right)$.}
\label{fig7}
\end{figure}
\begin{lemma} \label{lemma2}
Interconnection network $\hat{\mathcal{S}}\left(n\times r\right)$ is equivalent to shuffle network $\mathcal{S}(N)$ in terms of connectivity, where $N=nr$.
\end{lemma}
\begin{proof}
Consider fiber $\hat{f}\left(ap'bq',bq'ap'\right)$ in $\hat{\mathcal{S}}\left(n\times r\right)$. If $p=ar+p'$ and $q=br+q'$, we have the following one-to-one correspondence:
\begin{equation}
p,q\leftrightarrow ap',bq'\leftrightarrow \hat{f}\left(ap'bq',bq'ap'\right)
\end{equation}
According to properties P1 and P2 of $\mathcal{S}(N)$, there is also a one-to-one correspondence:
\begin{equation}
p,q\leftrightarrow f\left(pq,qp\right)
\end{equation}
Therefore, there is a one-one and onto mapping between $\hat{f}\left(ap'bq',bq'ap'\right)$ in $\hat{\mathcal{S}}\left(n\times r\right)$ and $f(pq,qp)$ in $\mathcal{S}(N)$:
\begin{equation}
\hat{f}\left(ap'bq',bq'ap'\right)\leftrightarrow f\left(pq,qp\right)
\end{equation}
This clearly demonstrates the equivalence between $\hat{\mathcal{S}}\left(n\times r\right)$ and $\mathcal{S}(N)$ in terms of connectivity.
\end{proof} 
\noindent We thus call $\hat{\mathcal{S}}\left(n\times r\right)$ modular shuffle network.
\begin{figure}[bp]
\centering
\includegraphics[scale=0.9]{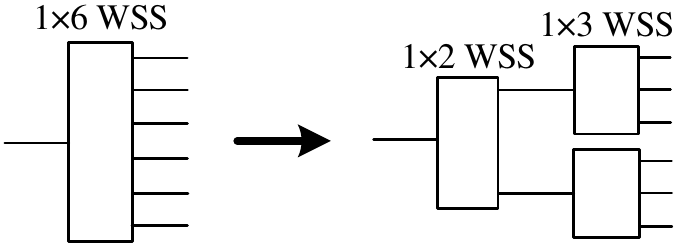}
\caption{A $1\times6$ WSS is equivalent to a $1\times2$ WSS cascaded by 2 $1\times3$ WSSs.}
\label{fig8}
\end{figure}

We replace $\mathcal{S}(N)$ with $\hat{\mathcal{S}}\left(n\times r\right)$ in OXC $\mathcal{Q}(N,w)$, and obtain a new switching network, called $\mathcal{Q}'\left(n\times r,w\right)$. For example, Fig. 7 is the switching network called $\mathcal{Q}'\left(2\times 3,3\right)$, which is produced by replacing shuffle network $\mathcal{S}(6)$ in Fig. \ref{fig1} with modular shuffle network $\hat{\mathcal{S}}\left(2\times3\right)$ in Fig. \ref{fig6}.

As Fig. \ref{fig7} illustrates, the $p$th $1\times N$ WSS that connects with input group $p$ of $\mathcal{S}(N)$ in $\mathcal{Q}(N,w)$ now connects with input group $ap'$ of $\hat{\mathcal{S}}\left(n\times r\right)$ in $\mathcal{Q}'\left(n\times r,w\right)$, where $p=ar+p'$. We thus relabel the $1\times N$ WSS as $ap'$. Similarly, we relabel $N\times1$ WSS $q$ as $bq'$, where $q=br+q'$. For example, $1\times N$ WSS 3 that connects with input group 3 in Fig. \ref{fig1} now connects with input group 10 of $\hat{\mathcal{S}}\left(2\times3,3\right)$ in Fig. \ref{fig7}, and thus is relabeled as 10. Similarly, $N\times1$ WSS 2 in Fig. \ref{fig1} is relabeled as 02.
\subsection{WSS Decomposition}
As \cite{17} mentions, a $1\times N$ WSS is functionally equivalent to a $1\times n$ WSS, each output of which is cascaded by a $1\times r$ WSS, where $N=nr$. Fig. \ref{fig8} illustrates that a $1\times 6$ WSS can be replaced by a $1\times 2$ WSS cascaded by 2 $1\times 3$ WSSs.
\begin{figure}[tp]
\centering
\includegraphics[scale=0.59]{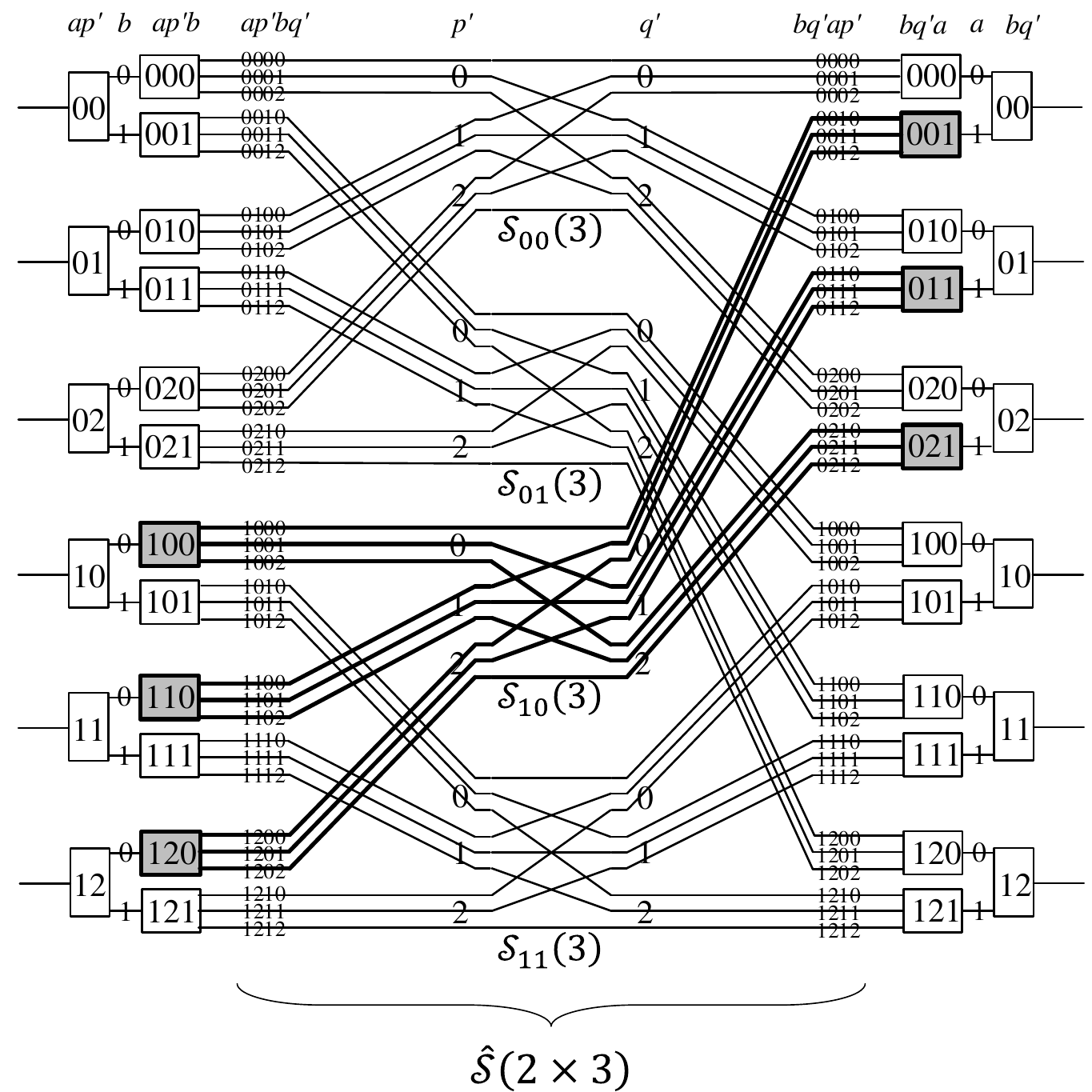}
\caption{A $6\times6$ switching network $\mathcal{Q}''\left(2\times3,w\right)$.}
\label{fig9}
\end{figure}

We replace each $1\times N$ WSS in $\mathcal{Q}'\left(n\times r,w\right)$ by a $1\times n$ WSS cascaded by $n$ $1\times r$ WSSs, and obtain a new switching network $\mathcal{Q}''\left(n\times r,w\right)$. After the replacement, $\mathcal{Q}'\left(2\times3,w\right)$ in Fig. \ref{fig7} changes to $\mathcal{Q}''\left(2\times3,w\right)$ in Fig. \ref{fig9}.

As Fig. \ref{fig9} plots, the $1\times n$ WSSs and the $n\times1$ WSSs become the inputs and the outputs of $\mathcal{Q}''\left(n\times r,w\right)$. We label the $1\times n$ WSS decomposed from $1\times N$ WSS $ap'$ in $\mathcal{Q}'\left(n\times r,w\right)$ as $ap'$, and the $1\times r$ WSS attached at the $b$th output of $1\times n$ WSS $ap'$ as $ap'b$, where $b=0,1,\cdots,n-1$. We label the $n\times1$ WSSs at the outputs in the same way. Clearly, $1\times r$ WSS $ap'b$ connects with input sub-group $ap'b$ of $\hat{\mathcal{S}}\left(n\times r\right)$ and thus input group $p'$ of $\mathcal{S}_{ab}(r)$. Similarly, $r\times1$ WSS $bq'a$ connects with output sub-group $bq'a$ of $\hat{\mathcal{S}}\left(n\times r\right)$ and thus output group $q'$ of $\mathcal{S}_{ab}(r)$. For example, $1\times2$ WSS 10 in Fig. \ref{fig9} is decomposed from $1\times6$ WSS $ap'=10$ in Fig. \ref{fig7}. The $1\times 3$ WSS at the 0th output of $1\times 2$ WSS 10 is labelled as $ap'b=100$, and it connects with input sub-group 100 of $\hat{\mathcal{S}}\left(2\times3\right)$ and thus input group $p'=0$ of $\mathcal{S}_{10}(3)$.

\subsection{Merge Operation}
According to the description in Section \ref{construction}-B, a $r^2\times r^2$ shuffle subnetwork $\mathcal{S}_{ab}(r)$ connects with $1\times r$ WSSs $a0b,a1b,\cdots,a(r-1)b$ and $r\times1$ WSSs $b0a,b1a,\cdots,b(r-1)a$. It is interesting to see that these components actually constitute a $r\times r$ OXC $\mathcal{Q}(r,w)$. For instance, $\mathcal{S}_{10}(3)$, $1\times3$ WSSs 100,110,120, and $3\times1$ WSSs 001,011,021 in Fig. \ref{fig9} composes a $3\times3$ OXC $\mathcal{Q}(3,3)$.

We thus perform $n^2$ merge operations. In the $k$-th merge, we combine $r^2\times r^2$ shuffle subnetwork $\mathcal{S}_{ab}(r)$, $1\times r$ WSSs $a0b,a1b,\cdots,a(r-1)b$ and $r\times1$ WSSs $b0a,b1a,\cdots,b(r-1)a$ to a $r\times r$ OXC, labelled by $\mathcal{Q}_{ab}(r)$, where $k=an+b$. It is clear that $1\times r$ WSSs $ap'b$ is now the $p'$th input of $\mathcal{Q}_{ab}(r)$ and $r\times1$ WSSs $bq'a$ is the $q'$th output of $\mathcal{Q}_{ab}(r)$. As a result, we obtain a modular $N\times N$ OXC, which is actually an interconnection network of $N$ $1\times n$ WSSs, $n^2$ $r\times r$ OXCs, and $N$ $n\times1$ WSSs. We denote this modular OXC as $\hat{\mathcal{Q}}\left(n\times r,w\right)$. For example, after $2^2=4$ merge operations, $\mathcal{Q}''\left(2\times3,3\right)$ changes to $\hat{\mathcal{Q}}\left(2\times 3,3\right)$ in Fig. 10. Also, $1\times3$ WSSs 100 and $3\times1$ WSS 021 are now the 0th input and the 2nd output of $\mathcal{Q}_{10}(3)$, respectively.
\begin{figure}
\centering
\includegraphics[scale=0.65]{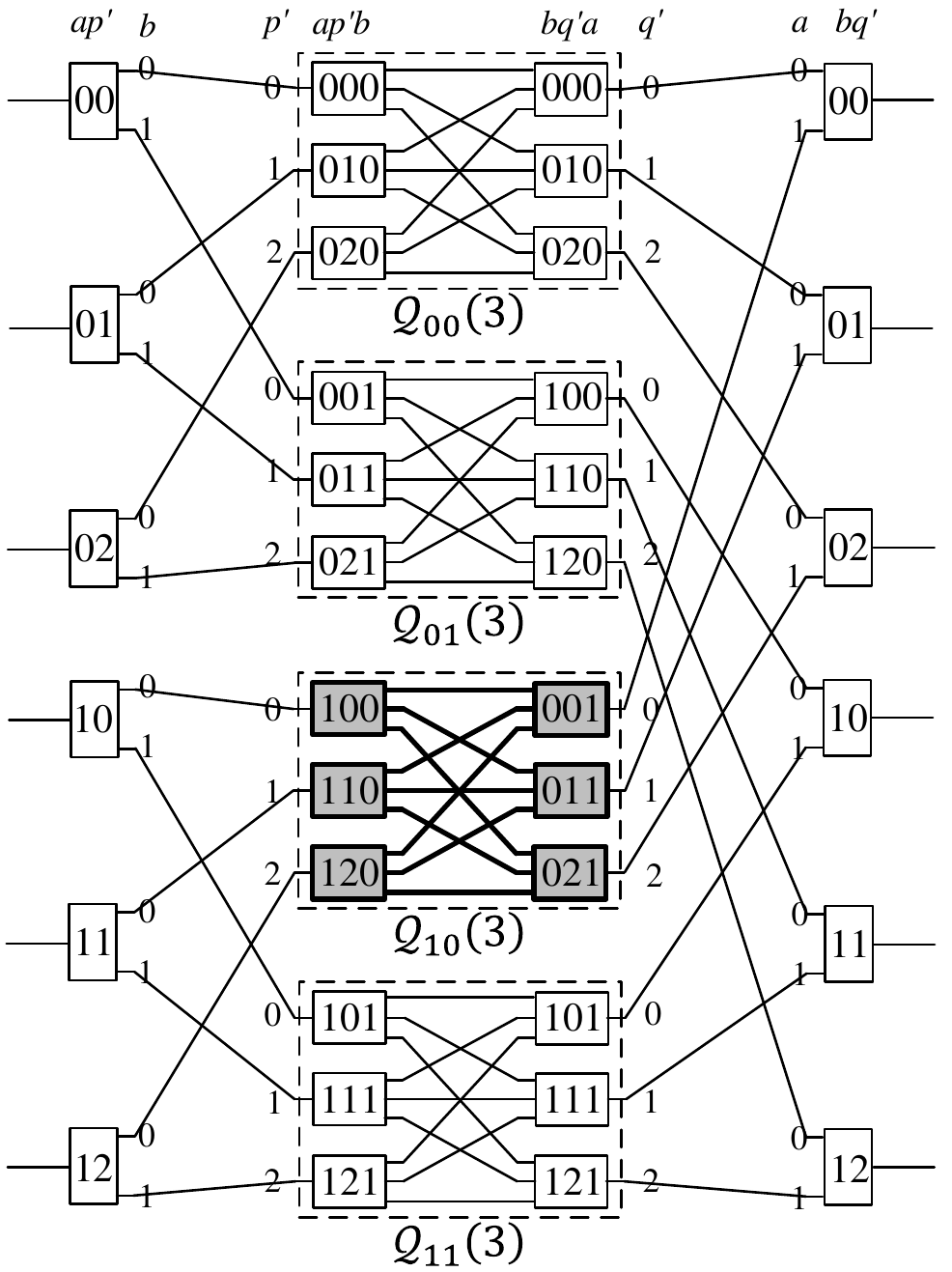}
\caption{A $6\times6$ modular OXC $\hat{\mathcal{Q}}\left(2\times3,w\right)$.}
\label{fig10}
\end{figure}
\begin{figure}[bp]
\centering
\includegraphics[scale=0.65]{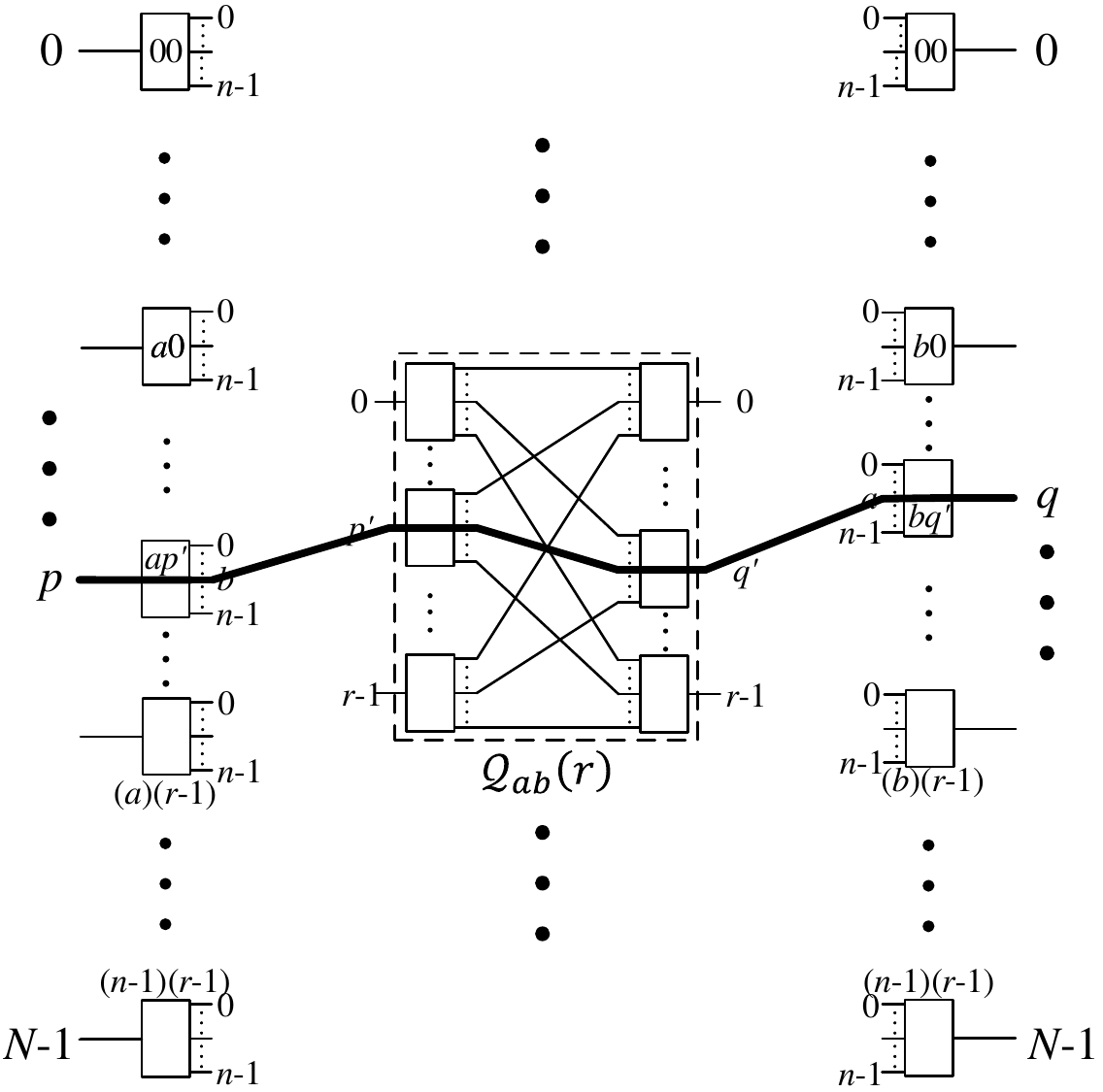}
\caption{An $N\times N$ modular OXC $\hat{\mathcal{Q}}\left(n\times r,w\right)$.}
\label{fig11}
\end{figure}

Fig. \ref{fig11} plots an $N\times N$ modular OXC $\hat{\mathcal{Q}}\left(n\times r,w\right)$, where $N=nr$. In particular, the connectivity of $\hat{\mathcal{Q}}\left(n\times r,w\right)$ is as follows:
\begin{enumerate}[1)]
\item The $p$th $1\times n$ WSS at the input stage is labelled by $ap'$, and the $q$th $n\times1$ WSS at the output stage is labelled by $bq'$, where $a,b=0,1,\cdots,n-1$, $p',q'=0,1,\cdots,r-1$, $p=ar+p'$ and $q=br+q'$;
\item The $b$th output of $1\times n$ WSS $ap'$ connects with the $p'$th input of $\mathcal{Q}_{ab}(r)$;
\item The $a$th input of $n\times 1$ WSS $bq'$ connects with the $q'$th output of $\mathcal{Q}_{ab}(r)$;
\end{enumerate}
\section{Performance Evaluation} \label{evaluation}
In this section, we evaluate the performance of $N\times N$ OXC $\hat{\mathcal{Q}}\left(n\times r,w\right)$, where $N=nr$. We analyze the routing property in Section \ref{evaluation}-A, simulate the physical-layer performance in Section \ref{evaluation}-B, and discuss the scalability in Section \ref{evaluation}-C.
\subsection{Nonblocking Routing Property}
Let $\hat{R}\left(ap',bq',\lambda_i\right)$ be a connection at wavelength $\lambda_i$ from $1\times n$ WSS $ap'$ to $n\times1$ WSS $bq'$ in $\hat{\mathcal{Q}}\left(n\times r,w\right)$, where $i=0,1,\cdots,w-1$. According to the connectivity of $\hat{\mathcal{Q}}\left(n\times r,w\right)$ in Section \ref{construction}-C, it is easy to know that the path of $\hat{R}$ is uniquely determined by the input address $ap'$ and the output address $bq'$ as follows:
\begin{align}
&\text{ Input of } 1\times n \text{ WSS } ap'\nonumber\\ 
\rightarrow& \text{ output } b \text{ of } 1\times n \text{ WSS } ap'\nonumber\\ 
\rightarrow& \text{ input } p' \text{ of } \mathcal{Q}_{ab}(r) \nonumber\\ 
\rightarrow& \text{ output } q' \text{ of } \mathcal{Q}_{ab}(r)\nonumber\\ 
\rightarrow& \text{ input } a \text{ of } n\times1 \text{ WSS } bq' \nonumber\\ 
\rightarrow& \text{ output of } n\times1 \text{ WSS } bq'.\nonumber
\end{align}
This indicates that $\hat{\mathcal{Q}}\left(n\times r,w\right)$ has a self-routing property similar to that of the classical OXC.

Since the $1\times n$ WSSs or $1\times r$ WSSs can switch different wavelengths independently, the connections at different wavelengths do not conflict with each other. In the following, we show that the connections at the same wavelength will not suffer internal contentions in $\hat{\mathcal{Q}}\left(n\times r,w\right)$.
\newtheorem{theorem}{Theorem}
\begin{theorem}\label{theorem1}
The modular OXC $\hat{\mathcal{Q}}\left(n\times r,w\right)$ is nonblocking at each wavelength.
\end{theorem}
\begin{proof}
Consider an extreme case, where wavelength $\lambda$ is idle only at input $ap'$ and output $bq'$. In the following, we show that $\hat{R}\left(ap',bq',\lambda\right)$ can be built up in this case.

Wavelength $\lambda$ is idle at input $ap'$, which implies that $\lambda$ is idle at all the outputs of $1\times n$ WSS. Similarly, $\lambda$ is idle at all the inputs of $n\times1$ WSS $bq'$. According to the connectivity of $\hat{\mathcal{Q}}\left(n\times r,w\right)$, $\lambda$ is free at both the link between output $b$ of $1\times n$ WSS $ap'$ and input $p'$ of $\mathcal{Q}_{ab}(r)$, and that between input $a$ of $n\times1$ WSS $bq'$ and output $q'$ of $\mathcal{Q}_{ab}(r)$. Lemma \ref{lemma1} shows that $\mathcal{Q}_{ab}(r)$ can set up a sub-connection $R\left(p',q',\lambda \right)$ from its input $p'$ to its output $q'$, as long as $\lambda$ is idle at its input $p'$ and its output $q'$. Thus, $\hat{R}$ can be established, which proves this theorem.
\end{proof}
\subsection{Physical-layer Performance}
We use a commercial software OptiSystem to simulate the transmission performance of a $64\times64$ modular OXC $\hat{\mathcal{Q}}\left(8\times8,w\right)$, which contains 64 $1\times8$ WSSs, 64 $8\times8$ $\mathcal{Q}(8,w)$, and 64 $8\times1$ WSSs. As Fig. \ref{fig11} plots, a connection will pass through two $1\times8$ WSSs and two $8\times1$ WSSs. Since there is no $1\times8$ WSS module in OptiSystem, we construct the $1\times8$ WSS using 1 $1\times8$ DeMux, 8 $1\times8$ OSs, and 8 $8\times1$ Muxs, according to the structure of PLC-based WSSs in \cite{14}.
\begin{figure}
\centering
\subfigure[Simulation setup]{
 \label{fig12a}
 \includegraphics[scale=0.71]{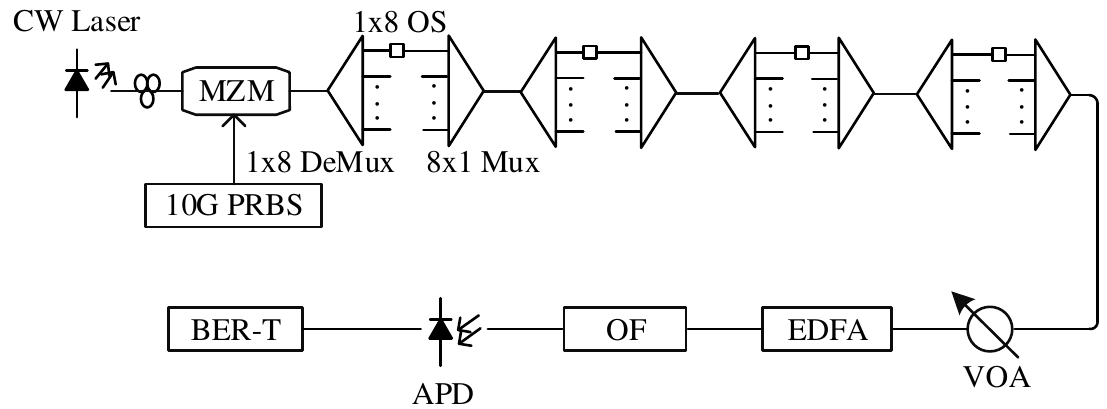}}
\subfigure[Simulation results]{
 \label{fig12b}
 \includegraphics[scale=0.75]{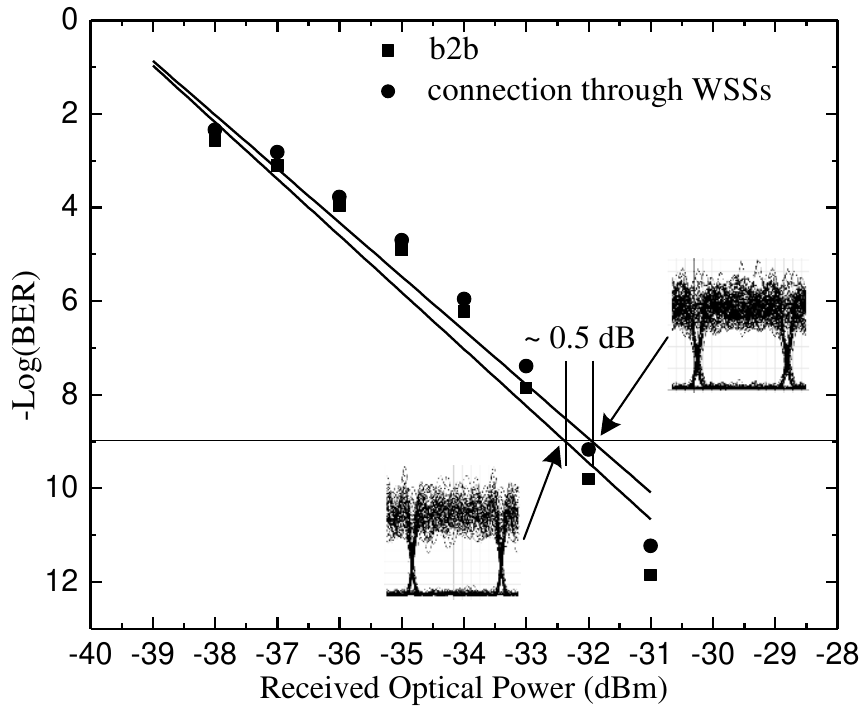}}
\caption{Transmission performance of an $8\times8$ modular OXC.}
\label{fig12}
\end{figure}

Fig. \ref{fig12a} shows the simulation setup. We generate a 10-Gb/s optical signal at wavelength 1552.53 nm using a Mach-Zehnder Modulator (MZM), and inject it into a path containing 4 $1\times8$ DeMuxs, 4 $8\times1$ OSs, and 4 $8\times1$ Muxs. At the receiver side, the optical signal goes through a variable optical attenuator (VOA) for power adjustment, an erbium-doped fiber amplifier (EDFA) for power amplification, and an optical filter (OF) for spontaneous emission noise suppression. Via an avalanche photodiode (APD), the optical signal is converted to an electrical signal, which is finally fed to a bit-error rate tester (BERT) to measure bit-error-ratio (BER) performance.

In the simulation, we measure the BER of the connection under different received optical powers (ROPs) by tuning the VOA. We compare the BER performance of the connection with that of a back-to-back (BtB) signal that only passes through a short fiber link in Fig. \ref{fig12b}. Compared to the BtB signal, the connection only suffers a power penalty of $\sim0.5$ dB at the BER of $10^9$. Furthermore, as \cite{18} shows, the insertion loss of a $1\times k$ LCOS-based WSS is $\sim5$ dB, which is independent of the port count $k$. Thus, the insertion loss of $\hat{\mathcal{Q}}\left(n\times r,w\right)$ is $\sim20$ dB, and will not increase with the port count of the OXC. This indicates that the modular OXC is scalable in terms of the transmission performance and the insertion loss.

To be cost-effective, it is possible to replace each $1\times n$ WSS at the input with a $1\times n$ OC, which however may increase the insertion loss of the modular OXC. For example, if we replace each $1\times8$ WSS at the input of $\hat{\mathcal{Q}}\left(8\times8,w\right)$, the insertion loss will be increased by 4 dB.
\subsection{Cabling Complexity}
The modularization of OXCs not only scales down the size of employed optical components, but also cuts down the cabling complexity. As Fig. \ref{fig11} shows, the number of fibers between two adjacent stages is $2Nn$, which is $2/r$ of that of the classical OXC. For example, a $160\times160$ classical OXC contains $160^2=25600$ fibers. The number of fibers is reduced to 2560, if we use $1\times8$ WSSs and $20\times20$ OXCs to construct a $160\times160$ modular OXC $\hat{\mathcal{Q}}\left(8\times20,w\right)$. In particular, the number of cables is $2N^{1.5}$ when $r=n=\sqrt{N}$.

Also, all the fibers inside each $r\times r$ OXC can be capsulated in a standard module, and thus will not increase the cabling complexity of the entire network. An example of the integrated OXC module is the $5\times5$ WSS reported by \cite{19}. We expect that the size of OXC module will be large in the near future.
\section{Conclusion} \label{conclusion}
In this paper, we propose a method to construct a large-scale WSS-based OXC. Our analysis shows that the proposed OXC has the following advantages: (1) it is modular since it can be constructed from a set of small-size WSSs and small-size OXC modules; (2) it is scalable because its cabling complexity is smaller than that of the classical OXC; (3) its power penalty is small, and its insertion loss is acceptable and does not increase with the port count.
\ifCLASSOPTIONcaptionsoff
  \newpage
\fi
\bibliographystyle{IEEEtran}
\bibliography{IEEEabrv,Ref}







\end{document}